\documentclass[letterpaper,10pt,conference]{ieeeconf}
 
\IEEEoverridecommandlockouts
\usepackage{graphicx}
\usepackage{xcolor}
\usepackage{subcaption}
\usepackage{amsmath}
\usepackage{amssymb}
\usepackage{mathtools}
\usepackage{algorithm}
\usepackage{algorithmic}
 
\mathtoolsset{showonlyrefs}

\newtheorem{definition}{Definition}

\newtheorem{theorem}{Theorem}
\newtheorem{corollary}{Corollary}

\newtheorem{remark}{Remark}

\newtheorem{problem}{Problem}
 
\title{\LARGE \bf
{Fast Risk Certification of Candidate Trajectories
under Uncertain Time-Varying Constraints
}}
 
\author{Srimanta Santra$^{1}$, Oleksii Molodchyk$^{1}$, Matti Noack$^{1}$ and Timm Faulwasser$^{1}$%
\thanks{$^{1}$Institute of Control Systems, Hamburg University of Technology,
21073 Hamburg, Germany.
E-mail: \{srimanta.santra, oleksii.molodchyk, matti.noack\}@tuhh.de, timm.faulwasser@ieee.org}%
\thanks{This project is funded by the Deutsche Forschungsgemeinschaft
(DFG, German Research Foundation) - SFB 1615 - 503850735.}%
}
 
\begin{document}
 
\maketitle
\thispagestyle{empty}
\pagestyle{empty}

\begin{abstract}
This paper studies the certification of a fixed candidate trajectory on a
finite certification grid under parametric uncertainty. For each
constraint-time pair, we define a scalar measure of constraint
violation and aggregate the resulting pointwise chance constraints into
a worst-case Value-at-Risk (VaR) margin. The goal is not to generate a
new trajectory, but to assess online whether a trajectory produced by a
planner or predictive controller is sufficiently safe on the
certification grid. Direct evaluation requires repeated uncertainty
propagation and is often too expensive for computationally demanding
models. We therefore adopt an offline-online scheme: offline, a
surrogate of the constraint violation map along the candidate
trajectory is constructed using polynomial chaos expansion (PCE) when
the uncertainty law is known, or kernel regression when only sampled
input-output data are available; online, the surrogate is sampled to
evaluate conservative VaR bounds at low computational cost. On the
theoretical side, we derive a finite-sample upper bound for the
grid-based VaR margin using empirical quantiles, the
Dvoretzky-Kiefer-Wolfowitz (DKW) inequality, and a union bound over all
constraint-time pairs, without assuming a parametric family for the
underlying violation distribution. We also show how a uniform
surrogate error bound transfers to the certified VaR margin. The
approach is illustrated on a crystallization population balance model,
where the surrogate-based risk estimates track direct Monte Carlo
results while substantially reducing online evaluation time.
\end{abstract}

\section{Introduction}
\label{sec:intro}
Trajectory planning, predictive control, and learning-based decision
schemes often generate candidate trajectories that are safe for a
nominal model. In safety-critical settings, however, nominal feasibility
is not sufficient: a candidate trajectory should also be certified under
uncertainty quickly enough for online use. This need arises, for example, in risk-aware trajectory verification and
chance-constrained planning under disturbances \cite{Blackmore2011},
model mismatch \cite{jasour2021realtime}, and uncertain environments
\cite{Wang2020}. It is particularly
relevant for systems with expensive dynamics, such as population balance
models and other PDE constrained systems, where repeated uncertainty
propagation is too costly for real-time decision support.

This paper studies safety \emph{certification}, rather than trajectory
generation. Given a fixed candidate trajectory over a finite horizon and
uncertain time-varying safety constraints, we ask whether the
probability of violation remains below a prescribed risk level at every
point of a certification grid. In contrast to methods that synthesize a
new control policy or trajectory, our goal is a fast and rigorous
\emph{post-hoc certificate} for a trajectory that has already been
generated.

The main challenge is computational. Direct evaluation requires repeated
propagation of the uncertain model and repeated estimation of tail
behavior across all constraint-time pairs. For nonlinear and
computationally intensive systems, this is generally too expensive for
online use. Related computational bottlenecks have motivated tractable
approximations in chance-constrained control and stochastic MPC
\cite{Mesbah2016,Farina2016}. Here, the difficulty is amplified by the
need to certify an entire finite-horizon trajectory while retaining a
clear risk interpretation.

Our approach combines offline surrogate construction with online risk
certification. Polynomial chaos expansion (PCE) is used when the uncertainty law 
is known and the constraint evaluation is sufficiently smooth, 
yielding a structured surrogate that can be evaluated cheaply 
\cite{xiu2002wiener}. Kernel-based regression provides a complementary
data-driven alternative when one prefers a flexible nonparametric model
or only sampled input-output data are available
\cite{aronszajnTheoryReproducingKernels1950,scholkopfGeneralizedRepresenterTheorem2001,Schobi2015}.
Using these surrogates, we summarize the family of pointwise chance
constraints by a single worst-case Value-at-Risk (VaR) margin on the
certification grid. VaR based formulations have recently attracted
considerable interest in risk-aware stochastic control and planning
\cite{Tooranjipour2024}. Closest in spirit is
\cite{Miller2025}, which studies upper bounds on the maximum VaR of a
state function along stochastic processes. Our setting differs in that
we focus on post-hoc certification of a fixed trajectory on a prescribed
grid and derive finite sample guarantees for that grid based problem.

The contributions of the paper are threefold: First, we formulate the
certification problem for a fixed candidate trajectory as a worst-case
VaR condition on a prescribed certification grid. Second, we derive a
finite-sample upper bound for this grid-based VaR margin using empirical
quantiles, the Dvoretzky-Kiefer-Wolfowitz inequality, and a union bound
over all constraint-time pairs. Third, we show how a uniform surrogate
approximation error propagates to the VaR margin, which yields a
sufficient surrogate-based certification rule. The resulting guarantee
is pointwise on the certification grid, it is not a joint chance
guarantee over the full horizon.

\section{Problem Formulation}
\label{sec:problem}
We study safety certification of a fixed candidate trajectory on a
prescribed discrete-time grid
$\mathcal{T}=\{t_j\}_{j=1}^{T}$,
with $t_1=t_0$ and $t_{T}=t_f$.
This reflects the sampled nature of digital planning and control, where
candidate trajectories are generated, updated, and checked at discrete
decision instants. Accordingly, the proposed certificate is formulated
over $\mathcal T$. All guarantees established below are therefore
grid-based and pointwise-in-time. Continuous-time guarantees between
grid points are outside the scope of the present work.

All random quantities are defined on a common probability space
$(\Omega,\mathcal{F},\mathbb{P})$ with the space of outcomes $\Omega$, $\sigma$-algebra $\mathcal{F}$, and the probability measure $\mathbb{P}$. Capital letters denote random
variables, whereas lowercase letters denote deterministic realizations.
Let
$P:\Omega\to \Pi \subseteq \mathbb{R}^{n_p}$
be a time-invariant random parameter vector with probability law
$\mu_P$ and with realizations in $\Pi\subseteq\mathbb{R}^{n_p}$. 
We assume that $P$ admits a finite variance, and indicate it using the notation
$P\in L^2(\Omega,\mathcal{F},\mathbb{P};\mathbb{R}^{n_p})$.

Let
$x_c:\mathcal{T}\to\mathbb{R}^{n_x}$
denote the given deterministic candidate trajectory. Define $N$ maps
$g_i:[t_0,t_f]\times\mathbb{R}^{n_x}\times\mathbb{R}^{n_p}\to\mathbb{R}$ for each index $i\in\{1,\dots,N\}$. Consider constraints
$
g_i(t, x_c(t), p) \leq 0$, 
$\forall t \in \mathcal{T}
$,
imposed on $x_c$. Observe that the uncertainty here enters as a realization $p$ of $P\sim \mu_P$ through the last argument of each $g_i$.
 
Since $x_c$ is deterministic, the only source of uncertainty in evaluating each constraint $g_i$ is the random parameter $P$. Our goal is to reason probabilistically about constraint satisfaction, so
we lift each constraint evaluation to a scalar random variable
\begin{align}
\label{eq:violation_variable}
V_{i,j} &:=
g_i\bigl(t_j,\,x_c(t_j),\,
P\bigr),
\quad V_{i,j}:\Omega\to\mathbb{R}.
\end{align}
Intuitively, $V_{i,j}(\omega)$ is the constraint value when $p=P(\omega)$: negative values indicate satisfaction, positive values indicate violation.
We assume $V_{i,j}\in L^2(\Omega,\mathcal{F},\mathbb{P}; \mathbb{R})$ for all $(i,j)$.
This separates the
\emph{deterministic} structure of the constraint value function its dependence on $t$ and $p$ from its
\emph{random} realization through~$P$.
Typical examples include a motion planning module that proposes a fixed
lane-change or obstacle avoidance trajectory for an autonomous vehicle
\cite{Blackmore2011}, and a guidance module that proposes a fixed
reference path for a UAV. In both cases, the candidate trajectory is
deterministic, but the associated constraint violation variables become
random when evaluated under uncertain obstacle motion, wind
disturbances, or uncertain model parameters \cite{luders2016wind}. The
proposed framework is intended as a post-processing safety filter: a
planner or predictive controller first generates the candidate
trajectory $x_c$, and the certificate then checks whether $x_c$ is
sufficiently safe under uncertainty before implementation.
\begin{definition}[Value-at-Risk]
\label{def:VAR}
Let $Z$ be a scalar random variable and let $\delta\in(0,1)$.
The Value-at-Risk (VaR) of $Z$ at level $\delta$ is defined by
$\mathrm{VaR}_{\delta}(Z)
= \inf\bigl\{\alpha\in\mathbb{R}:\mathbb{P}(Z\le\alpha)\ge 1-\delta\bigr\}$.
Equivalently, for any $\alpha\in\mathbb{R}$,
\begin{align}
  \mathrm{VaR}_{\delta}(Z) \le \alpha
  \quad\Longleftrightarrow\quad
  \mathbb{P}(Z > \alpha) \le \delta.
  \label{eq:var_tail_equiv}
\end{align}
\vspace{-3.5ex} 
\end{definition}
Applying Definition~\ref{def:VAR} to the constraint value random variable
$V_{i,j}$, the condition
$\mathrm{VaR}_{\delta}(V_{i,j})\le 0$
is equivalent to
$\mathbb{P}(V_{i,j}>0)\le \delta$,
that is, the probability of violating constraint $i$ at time $t_j$
is at most $\delta$.
\begin{definition}[Worst-case grid VaR margin]
\label{def:rho_delta}
For a given candidate trajectory $x_c$, certification grid
$\mathcal T$, and risk level $\delta\in(0,1)$, define
$
  \rho_\delta
  = \max_{i\in\{1,\dots,N\}}
    \max_{j\in\{1,\dots,T\}}
    \mathrm{VaR}_{\delta}(V_{i,j}).
$
\end{definition}
\begin{remark}
\label{rem:pointwise_violation_bound}
The scalar quantity $\rho_\delta$ summarizes the worst-case pointwise
tail risk over all constraint-time pairs on the certification grid. In
particular, if
$\rho_\delta\le 0$,
then
$\mathbb{P}(V_{i,j}>0)\le \delta$,
$\forall i\in\{1,\dots,N\},\ \forall j\in\{1,\dots,T\}$.
Thus, $\rho_\delta\le 0$ provides a pointwise chance guarantee for
every constraint at every grid point. By a union bound over all $(i,j)$ constraint-time pairs, the
probability that at least one constraint is violated satisfies
$\mathbb{P}\bigl(\exists (i,j):V_{i,j}>0\bigr)
\le
\min\{1,N T\,\delta\}$.
This bound is generally conservative and should not be interpreted as a
joint chance constraint over the full horizon.
\end{remark}
\begin{problem}
\label{prob:main_problem}
Given the candidate trajectory $x_c$ on the certification grid
$\mathcal T$, the constraint violation $\{g_i\}_{i=1}^{N}$,
the induced constraint value $V_{i,j}$ defined in \eqref{eq:violation_variable},
the uncertainty law of $P$, and a prescribed risk level
$\delta\in(0,1)$, determine whether $\rho_\delta\le 0$.
The objective is to construct a certification procedure that can be
prepared offline and evaluated online at low computational cost.
\end{problem}

\section{VaR Evaluation}
\label{sec:method}
Problem~\ref{prob:main_problem} identifies $\rho_\delta$ as the central
quantity to certify. Its direct computation requires repeated
uncertainty propagation at each grid point and is generally too
expensive for online use. We therefore adopt a two-stage approach. In
the \emph{offline stage}, surrogates of the constraint violation random
variables are constructed from model evaluations on $\mathcal{T}$. In
the \emph{online stage}, these surrogates are queried to evaluate the
pointwise VaR values and the worst-case margin at negligible cost. To
this end, we construct surrogates for the random variables $V_{i,j}$, as
described in Section~\ref{subsec:pce} or
Section~\ref{subsec:kernel_extension}.

\subsection{Polynomial Chaos Surrogate of Constraint Violations}
\label{subsec:pce}
Polynomial chaos expansion (PCE) provides a structured surrogate for the
constraint value random variables $V_{i,j}$ defined
in~\eqref{eq:violation_variable}. Once constructed offline, it enables
fast sampling and efficient empirical VaR evaluation online.
Following \cite{ou2025polynomial}, we use $P$ directly as the argument
of the polynomial basis. Let $\{\Psi_k\}_{k=0}^\infty$ be a complete
orthonormal polynomial basis of $L^2(\Omega,\mathcal{F},\mathbb{P};
\mathbb{R})$, with $\Psi_0\equiv 1$,
satisfying
$\mathbb{E}\!\left[\Psi_k(P)\,\Psi_r(P)\right] = \delta_{kr}$,
$k,r\in\mathbb{N}_0$.
The choice of polynomial family depends on the distribution of $P$:
Hermite polynomials correspond to Gaussian measures and Legendre
polynomials to uniform measures
\cite{lemaitre2010spectral,ou2025polynomial}.

Every $V_{i,j}$ admits a unique $L^2$-convergent expansion
\begin{align}
V_{i,j}
=
\sum_{k=0}^{\infty} c_{i,j,k}\,\Psi_k(P),
\label{eq:pce_expansion}
\end{align}
with deterministic coefficients $c_{i,j,k}\in\mathbb R$ given by
\begin{align}
c_{i,j,k}
&=
\int_{\Pi}
g_i\bigl(t_j,x_c(t_j),p\bigr)\,\Psi_k(p)\,d\mu_P(p).
\label{eq:pce_coeff_exact}
\end{align}
Here $p\in\Pi$ is a deterministic integration variable, whereas
$P$ denotes the random parameter vector.
In practice, the infinite series \eqref{eq:pce_expansion} is truncated
to $K$ terms, yielding the surrogate
\begin{align}
\widehat V_{i,j}
=
\sum_{k=0}^{K-1} c_{i,j,k}\,\Psi_k(P).
\label{eq:pce_violation}
\end{align}
For a total-degree basis of maximal degree $d$ in $n_p$ stochastic
dimensions, the number of retained basis functions is
$
K=\binom{n_p+d}{d}.
$
Because the basis is orthonormal, the truncated expansion directly
provides moments
$\mathbb E[\widehat V_{i,j}] = c_{i,j,0}$,
$\mathrm{Var}(\widehat V_{i,j})
=
\sum_{k=1}^{K-1} c_{i,j,k}^2$.
In general, the coefficients in \eqref{eq:pce_coeff_exact} are not
available in closed form. We therefore compute them offline by
non-intrusive spectral projection
\cite{xiu2002wiener,lemaitre2010spectral}. Given $Q$ quadrature nodes
$\{p_q\}_{q=1}^{Q}\subset\Pi$ and corresponding weights
$\{w_q\}_{q=1}^{Q}$, we approximate
\begin{align}
c_{i,j,k}
&=
\sum_{q=1}^{Q}
w_q\,g_i\bigl(t_j,x_c(t_j),p_q\bigr)\,\Psi_k(p_q),
\label{eq:pce_projection}
\end{align}
for $k=0,\dots,K-1$. Alternative non-intrusive constructions, such as
regression-based PCE or stochastic collocation, can also be used
\cite{lemaitre2010spectral,tan2025offset,ou2025polynomial}.

Once the coefficients $c_{i,j,k}$ have been computed offline, the
surrogate $\widehat V_{i,j}$ in \eqref{eq:pce_violation} can be
evaluated at negligible cost for new realizations of $P$. This makes
repeated empirical quantile and VaR evaluation much cheaper than
repeated evaluation of the original constraint function $g_i$. PCE is particularly
effective when the dependence of $V_{i,j}$ on $P$ is sufficiently
smooth, so that low-order polynomials already provide an accurate
approximation \cite{xiu2002wiener,lemaitre2010spectral}.
It is worth emphasizing that the PCE construction itself is not
distribution-free, since the basis $\Psi_k$ is chosen according
to the probability law $\mu_P$ of the uncertain parameter vector.
The term distribution-free in the certification results below refers
instead to the finite-sample VaR bound, which does not require a
parametric family for the distribution of the constraint values $V_{i,j}$.


\subsection{Kernel-Based Surrogate of the Constraint Violation}
\label{subsec:kernel_extension}
When the dependence of the constraint values on the uncertain
parameter is not well captured by a fixed polynomial basis, or when 
the model is available only through sampled input-output evaluations, 
we use a kernel-based surrogate. Fix a constraint index 
$i\in\{1,\dots,N\}$ and define the joint input domain
$\mathcal{Z}:=\mathcal{T}\times\Pi$.
For each fixed $i$, we construct a single surrogate over the joint
input $z=(t,p)\in\mathcal{Z}$. This avoids training separate models
for every grid point $t_j$ and exploits regularity of the constraint
function $g_i$ across the certification grid.

The offline training data are obtained by evaluating the constraint
function $g_i$ at sampled time-parameter pairs. Specifically, let
$\mathcal{D}_i = \{(z_\ell, y_\ell)\}_{\ell=1}^{L}$,
$z_\ell = (t_\ell, p_\ell) \in \mathcal{Z}$,
$y_\ell = g_i\bigl(t_\ell,\, x_c(t_\ell),\, p_\ell\bigr)$,
where $t_\ell \in \mathcal{T}$ and $p_\ell \in \Pi$ denote the sampled
time and parameter values, respectively. Let
$\kappa: \mathcal{Z} \times \mathcal{Z} \to \mathbb{R}$
be a positive definite kernel
\cite{aronszajnTheoryReproducingKernels1950}. The kernel induces a
reproducing kernel Hilbert space (RKHS) $\mathcal{H}$ of scalar-valued
functions on $\mathcal{Z}$, see also
\cite[Ch.~1]{berlinetReproducingKernelHilbert2004}. Intuitively,
$\kappa(z,z')$ measures similarity between two input pairs
$z=(t,p)$ and $z'=(t',p')$, and thereby determines the class of
functions that can be represented by the surrogate.

We compute the surrogate by kernel ridge regression
\begin{align}
\widetilde g_i
\in
\arg\min_{f\in\mathcal H}
\sum_{\ell=1}^{L}
\bigl(f(z_\ell)-y_\ell\bigr)^2
+\lambda\|f\|_{\mathcal H}^2,
\label{eq:krr_problem}
\end{align}
where $\lambda>0$ is a regularization parameter and $\lVert \cdot \rVert_\mathcal{H}:\mathcal{H}\to\mathbb{R}$ is the induced RKHS norm. By the representer
theorem \cite{scholkopfGeneralizedRepresenterTheorem2001}, every
minimizer of \eqref{eq:krr_problem} admits the finite-dimensional form
$
\widetilde g_i(z)
=
\sum_{\ell=1}^{L}\alpha_\ell\,\kappa(z,z_\ell)$,
$z\in\mathcal Z,
$
for suitable coefficients $\alpha_\ell\in\mathbb R$.
Let
$
\alpha=[\alpha_1,\dots,\alpha_{L}]^\top$,
$y=[y_1,\dots,y_{L}]^\top.
$
Define the kernel matrix
$
K\in\mathbb R^{L\times L},
\qquad
(K)_{\ell m}=\kappa(z_\ell,z_m),
\qquad
\ell,m=1,\dots,L.
$
Then the coefficient vector is obtained from
$
(K+\lambda I)\alpha=y,
$
that is,
$
\alpha=(K+\lambda I)^{-1}y.
$

Once the deterministic surrogate map $\widetilde g_i$ has been computed
offline, it induces the surrogate constraint random variable
$\widetilde V_{i,j}:\Omega\to\mathbb R$ with its realization
\[\widetilde V_{i,j}(\omega)
:=
\widetilde g_i\bigl(t_j,P(\omega)\bigr).\]
Hence, the kernel-based approximation enters the certification pipeline
through the composition of the deterministic surrogate $\widetilde g_i$
with the random parameter vector $P$.
Pointwise VaR values are then estimated by drawing samples from the law
of $P$, or by resampling from empirical uncertainty data when only
samples are available, and evaluating
$\widetilde g_i(t_j,P(\omega))$ repeatedly. This yields empirical
samples of $\widetilde V_{i,j}$, from which conservative quantiles and
worst-case VaR margins can be computed exactly as in the finite-sample
certification procedure developed below.

\begin{remark}
The kernel surrogate is constructed through deterministic RKHS
regression and does not require a prescribed polynomial basis of
the constraint values $V_{i,j}$. Distributional
information on $P$ is needed only at the stage of VaR evaluation, for
example through direct sampling from $\mu_P$ or resampling from
available data. Under additional assumptions, kernel ridge regression
also admits probabilistic interpretations. For example, under Gaussian
observation noise and an appropriate Gaussian prior, the estimator
obtained from \eqref{eq:krr_problem} can be interpreted as the posterior
mean of a Gaussian process regression model
\cite{rasmussen2006gaussian}. For more general stochastic interpretations,
including Wiener kernel regression, we refer to
\cite{faulwasser2025wiener,Schobi2015}.
\end{remark}


\section{Main Results}
We now derive the certification guarantees that justify the use of both
surrogate constructions in Problem~\ref{prob:main_problem}.

 \subsection{Finite-Sample Certificate for VaR Margin}
 \label{subsec:offline_var_certification}
We first develop a finite-sample certificate from samples of the
constraint values $V_{i,j}$ computed along the fixed candidate
trajectory on the certification grid. The same empirical construction
will later be applied to values generated by a surrogate model.

Let $P:\Omega\to\mathbb R^{n_p}$ be the uncertain parameter vector from
Section~\ref{sec:problem}, and let
$\omega_{1},\dots,\omega_{M_{\mathrm{off}}}\in\Omega$
be i.i.d.\ sample points. For each $m=1,\dots,M_{\mathrm{off}}$, define
the sampled parameter realization
$
p_{m}:=P(\omega_{m}).
$
For each constraint-time pair $(i,j)$, define the sampled constraint 
values
$
V_{i,j}(t_j,p_{m})
=
g_i\bigl(t_j,x_c(t_j),p_{m}\bigr),
$
$m=1,\dots,M_{\mathrm{off}}$.
Thus, $V_{i,j}(t_j,p_{m})$ is the $m$-th sample of the scalar random variable
$V_{i,j}=V_{i,j}(\omega)$.
Let
$
F_{i,j}(\alpha)=\mathbb P(V_{i,j}\le \alpha)
$
denote the true cumulative distribution function (CDF) of $V_{i,j}$, and
let
\begin{align}
\widehat F_{i,j}^{M_{\mathrm{off}}}(\alpha)
=
\frac{1}{M_{\mathrm{off}}}
\sum_{m=1}^{M_{\mathrm{off}}}
\mathbb{I}\!\left\{V_{i,j}(t_j,p_{m})\le \alpha\right\}
\label{eq:empirical_cdf}
\end{align}
be the corresponding empirical CDF.
To obtain a finite-sample confidence bound, we apply the
Dvoretzky-Kiefer-Wolfowitz (DKW) inequality together with a union
bound over all $NT$ constraint-time pairs
\cite{dvoretzky1956asymptotic,massart1990tight}. For a prescribed
failure probability $\beta\in(0,1)$, define
\begin{align}
\varepsilon_{M_{\mathrm{off}}}
=
\sqrt{\frac{1}{2M_{\mathrm{off}}}
\log\!\Bigl(\frac{2NT}{\beta}\Bigr)}.
\label{eq:epsN_def}
\end{align}
Then the event
\begin{align}
\mathcal E_{M_{\mathrm{off}}}
=
\Bigl\{
\sup_{\alpha\in\mathbb R}
\bigl|
\widehat F_{i,j}^{M_{\mathrm{off}}}(\alpha)-F_{i,j}(\alpha)
\bigr|
\le \varepsilon_{M_{\mathrm{off}}},
\ \forall i,\forall j
\Bigr\}
\label{eq:event_E}
\end{align}
satisfies
$
\mathbb P(\mathcal E_{M_{\mathrm{off}}})\ge 1-\beta.
$
On the event $\mathcal E_{M_{\mathrm{off}}}$, the empirical CDF may
underestimate the true CDF by at most
$\varepsilon_{M_{\mathrm{off}}}$. Therefore, a conservative empirical
quantile is obtained by replacing the target CDF level $1-\delta$ by
\begin{align}
p_{M_{\mathrm{off}}}^+
=
1-\delta+\varepsilon_{M_{\mathrm{off}}}.
\label{eq:pNplus}
\end{align}
To ensure that $p_{M_{\mathrm{off}}}^+$ is a valid CDF level, we assume
\begin{align}
\varepsilon_{M_{\mathrm{off}}}\le \delta.
\label{eq:epsN_gmall}
\end{align}
Then $p_{M_{\mathrm{off}}}^+\in[0,1]$, and for each pair $(i,j)$ we define
the conservative empirical quantile
\begin{align}
\bar q_{i,j,\delta}
=
\inf\bigl\{
\alpha\in\mathbb R:
\widehat F_{i,j}^{M_{\mathrm{off}}}(\alpha)\ge p_{M_{\mathrm{off}}}^+
\bigr\}.
\label{eq:conservative_var_def}
\end{align}
This condition is essential for the certificate below. If
\eqref{eq:epsN_gmall} is violated, then
$p_{M_{\mathrm{off}}}^+>1$, and the conservative quantile in
\eqref{eq:conservative_var_def} is not defined in the present form.
Then
$
\bar q_{i,j,\delta}
=
V_{i,j}^{[k^+]},
$
where
$
k^+=\bigl\lceil p_{M_{\mathrm{off}}}^+\,M_{\mathrm{off}}\bigr\rceil.
$
Thus, $\bar q_{i,j,\delta}$ is the conservative empirical
$(1-\delta)$-quantile of $V_{i,j}$ for the pair $(i,j)$.
The corresponding worst-case certificate is
$\bar\rho_\delta=\max_i\max_j \bar q_{i,j,\delta}$.
Hence, the certification procedure has two stages: first construct the
pairwise conservative quantiles $\bar q_{i,j,\delta}$, and then take
their maximum over all constraint-time pairs to obtain the single scalar
certificate $\bar\rho_\delta$.

\begin{remark}
\label{rem:beta}
The parameter $\beta$ is the failure probability of the certificate.
More precisely, the uniform CDF event
$\mathcal E_{M_{\mathrm{off}}}$ in \eqref{eq:event_E} holds with
probability at least $1-\beta$. A smaller value of $\beta$ therefore
yields a more reliable certificate, but also increases the radius
$\varepsilon_{M_{\mathrm{off}}}$ in \eqref{eq:epsN_def}, and hence
increases conservatism. 
\end{remark}

The certificate is distribution-free in the sense that it does not
assume any specific parametric form for the distribution of
$V_{i,j}$. It relies only on i.i.d.\ samples and the DKW inequality.
\begin{remark}
The PCE basis $\Psi_k$ is chosen according to $\mu_P$, so the
surrogate construction itself is not distribution-free. The term
\emph{distribution-free} in the certification results below refers
solely to the finite-sample VaR bound, which requires no parametric
assumption on the distribution of $V_{i,j}$.
\end{remark}

\begin{theorem}
\label{thm:offline_certificate}
Suppose \eqref{eq:epsN_gmall} holds. Then, with probability of at least
$1-\beta$, the pairwise conservative quantiles defined in
\eqref{eq:conservative_var_def} satisfy
\begin{align}
\mathrm{VaR}_\delta(V_{i,j}) \le \bar q_{i,j,\delta},
\qquad \forall i,\forall j.
\label{eq:offline_cert_pairwise}
\end{align}
Consequently,
\begin{align}
\rho_\delta \le \bar\rho_\delta.
\label{eq:offline_cert_grid}
\end{align}
In particular, if $\bar\rho_\delta\le 0$, then with probability at least
$1-\beta$,
\begin{align}
\mathbb P(V_{i,j}>0)\le \delta,
\qquad \forall i,\forall j.
\label{eq:offline_cert_tail}
\end{align}
\vspace{-3.5ex}
\end{theorem}

\begin{proof}
On the event $\mathcal E_{M_{\mathrm{off}}}$, for every pair $(i,j)$ and
every $\alpha\in\mathbb R$,
$
F_{i,j}(\alpha)
\ge
\widehat F_{i,j}^{M_{\mathrm{off}}}(\alpha)
-
\varepsilon_{M_{\mathrm{off}}}.
$
Now fix $(i,j)$ and set $\alpha=\bar q_{i,j,\delta}$. By
\eqref{eq:conservative_var_def},
$
\widehat F_{i,j}^{M_{\mathrm{off}}}(\bar q_{i,j,\delta})
\ge
p_{M_{\mathrm{off}}}^+
=
1-\delta+\varepsilon_{M_{\mathrm{off}}}.
$
Hence,
$
F_{i,j}(\bar q_{i,j,\delta})
\ge
\widehat F_{i,j}^{M_{\mathrm{off}}}(\bar q_{i,j,\delta})
-
\varepsilon_{M_{\mathrm{off}}}
\ge
1-\delta.
$
By Definition~\ref{def:VAR}, this implies
$
\mathrm{VaR}_\delta(V_{i,j})\le \bar q_{i,j,\delta},
$
which proves \eqref{eq:offline_cert_pairwise}. Taking the maximum over
all $(i,j)$ gives \eqref{eq:offline_cert_grid}. Finally,
\eqref{eq:offline_cert_tail} follows from \eqref{eq:var_tail_equiv}.
\end{proof}

\begin{corollary}[Conservatism]
\label{cor:sample_complexity}
Let $\varepsilon^*>0$ and $\beta\in(0,1)$. If
\begin{align}
M_{\mathrm{off}}
\ge
\frac{1}{2(\varepsilon^*)^2}
\log\!\Bigl(\frac{2NT}{\beta}\Bigr),
\label{eq:sample_complexity}
\end{align}
then the confidence radius in \eqref{eq:epsN_def} satisfies
$
\varepsilon_{M_{\mathrm{off}}}\le \varepsilon^*.
$
In particular, if $\varepsilon^*\le\delta$, then the conservative
quantile level in \eqref{eq:pNplus} satisfies
$
p_{M_{\mathrm{off}}}^+
=
1-\delta+\varepsilon_{M_{\mathrm{off}}}
\le
1-\delta+\varepsilon^*.
$
Hence, increasing $M_{\mathrm{off}}$ decreases
$\varepsilon_{M_{\mathrm{off}}}$ and therefore reduces the
conservatism of the certificate in
Theorem~\ref{thm:offline_certificate}.
\end{corollary}

\begin{proof}
From \eqref{eq:epsN_def}, the condition
$\varepsilon_{M_{\mathrm{off}}}\le\varepsilon^*$ is equivalent to
$
\frac{1}{2M_{\mathrm{off}}}
\log\!\Bigl(\frac{2NT}{\beta}\Bigr)
\le
(\varepsilon^*)^2,
$
which yields \eqref{eq:sample_complexity}. The bound on
$p_{M_{\mathrm{off}}}^+$ then follows directly from its definition.
\end{proof}

\begin{remark}
For fixed $M_{\mathrm{off}}$, the DKW radius
$\varepsilon_{M_{\mathrm{off}}}$ in \eqref{eq:epsN_def} is essentially
optimal among distribution-free uniform CDF error bounds
\cite{massart1990tight}. Thus, without additional assumptions on the
distribution of $V_{i,j}$, the conservatism of the certificate can only
be reduced gradually by increasing the sample size
$M_{\mathrm{off}}$.
\end{remark}


\subsection{Surrogate Approximation Error}
\label{subsec:surrogate_transfer}
The finite-sample certificate in
Section~\ref{subsec:offline_var_certification} accounts only for the
error caused by using finitely many samples. If the constraint values
$V_{i,j}$ are evaluated through a surrogate model, then one must also
account for the approximation error of the surrogate itself. The next
result shows how a uniform surrogate error bound affects the worst-case
VaR margin. For each pair $(i,j)$, let $V_{i,j}$ denote the true
constraint value, and let $S_{i,j}$ denote a surrogate approximation
of $V_{i,j}$. In the PCE case, we set
$
S_{i,j}=\widehat V_{i,j},
$
whereas in the kernel-based case, we set
$
S_{i,j}=\widetilde V_{i,j}.
$
Define the corresponding surrogate worst-case VaR margin by
$
\rho_\delta^S
=
\max_{i\in\{1,\dots,N\}}
\max_{j\in\{1,\dots,T\}}
\mathrm{VaR}_\delta(S_{i,j}).
$

\begin{theorem}
\label{thm:surrogate_transfer}
Suppose there exists $\varepsilon\ge 0$ such that
\begin{align}
|V_{i,j}-S_{i,j}|\le \varepsilon
\qquad \text{a.s.}, \qquad \forall i,\forall j.
\label{eq:surr_uniform_error}
\end{align}
Then, for every pair $(i,j)$,
\begin{align}
\mathrm{VaR}_\delta(V_{i,j})
\le
\mathrm{VaR}_\delta(S_{i,j})+\varepsilon.
\label{eq:var_transfer_pairwise}
\end{align}
Consequently,
\begin{align}
\rho_\delta
\le
\rho_\delta^S+\varepsilon.
\label{eq:rho_transfer_bound}
\end{align}
In particular, if
$
\rho_\delta^S\le -\varepsilon$,
then $\rho_\delta\le 0$, and hence
$
\mathbb P(V_{i,j}>0)\le \delta,
\,\, \forall i,\forall j.
$
\end{theorem}

\begin{proof}
Fix a pair $(i,j)$. By \eqref{eq:surr_uniform_error},
$
V_{i,j}\le S_{i,j}+\varepsilon
\qquad \text{a.s.}
$
Hence, for every $\alpha\in\mathbb R$,
$
\{S_{i,j}\le \alpha-\varepsilon\}
\subseteq
\{V_{i,j}\le \alpha\}.
$
Therefore,
$
\mathbb P(V_{i,j}\le \alpha)
\ge
\mathbb P(S_{i,j}\le \alpha-\varepsilon).
$
Now choose
$
\alpha=\mathrm{VaR}_\delta(S_{i,j})+\varepsilon.
$
By Definition~\ref{def:VAR},
$
\mathbb P(S_{i,j}\le \alpha-\varepsilon)
=
\mathbb P\bigl(S_{i,j}\le \mathrm{VaR}_\delta(S_{i,j})\bigr)
\ge 1-\delta.
$
Hence,
$
\mathbb P(V_{i,j}\le \alpha)\ge 1-\delta.
$
Applying Definition~\ref{def:VAR} again gives
\eqref{eq:var_transfer_pairwise}. Taking the maximum over all
$(i,j)$ yields \eqref{eq:rho_transfer_bound}. The final implication
follows from Remark~\ref{rem:pointwise_violation_bound}.
\end{proof}

\begin{remark}
\label{rem:surr_error_source}
The bound $\varepsilon$ in \eqref{eq:surr_uniform_error} is any uniform
surrogate error bound. For PCE surrogates, such a bound may be obtained
from approximation theory under suitable regularity assumptions. For
kernel-based surrogates, it may be estimated from an independent
validation set by a uniform error assessment. Theorem~\ref{thm:surrogate_transfer}
does not depend on the specific surrogate construction, only on the
validity of \eqref{eq:surr_uniform_error}.
\end{remark}

Combining the finite-sample certificate from
Section~\ref{subsec:offline_var_certification} with
Theorem~\ref{thm:surrogate_transfer} gives the following practical
result.

\begin{corollary}
\label{cor:combined_surrogate_certificate}
Assume \eqref{eq:surr_uniform_error} holds a.s., and suppose the
assumptions of Theorem~\ref{thm:offline_certificate} hold for the
surrogate samples. Let $\bar\rho_\delta^S$ denote the
conservative certificate obtained by applying the finite-sample
construction of Section~\ref{subsec:offline_var_certification} to the
surrogate samples. Then, with probability at least $1-\beta$,
\begin{align}
\rho_\delta
\le
\bar\rho_\delta^S+\varepsilon.
\label{eq:combined_cert_bound}
\end{align}
In particular, if
$
\bar\rho_\delta^S\le -\varepsilon,
$
then
$
\mathbb P(V_{i,j}>0)\le \delta,
\qquad \forall i,\forall j,
$
with probability at least $1-\beta$.
\end{corollary}

\begin{proof}
By Theorem~\ref{thm:offline_certificate}, applied to the surrogate
samples and its assumptions holding by hypothesis,
$
\rho_\delta^S\le \bar\rho_\delta^S
$
with probability at least $1-\beta$. By
Theorem~\ref{thm:surrogate_transfer}, since \eqref{eq:surr_uniform_error}
holds a.s.,
$
\rho_\delta\le \rho_\delta^S+\varepsilon.
$
Combining the two inequalities gives
\eqref{eq:combined_cert_bound}. 
\end{proof}

\subsection{Online VaR Evaluation}
\label{subsec:online_var_evaluation}
Once a surrogate and a uniform surrogate error bound
$\varepsilon_{\mathrm{surr}}$ playing the role of $\varepsilon$
in Theorem~\ref{thm:surrogate_transfer} are available from the
offline stage, online certification reduces to surrogate evaluation
and empirical quantile computation on the certification grid. No further evaluation of the constraint function $g_i$ along the
candidate trajectory is required.
For online evaluation, draw $M_{\mathrm{on}}$ i.i.d.\ sample points
$\omega_{1},\dots,\omega_{M_{\mathrm{on}}}\in\Omega$ and form the
corresponding parameter realizations
$
p_{s}:=P(\omega_{s})\in\mathbb{R}^{n_p}$
$s=1,\dots,M_{\mathrm{on}}$. For each constraint-time pair $(i,j)$,
define the surrogate samples
\begin{align}
v_{s}
=
\begin{cases}
\widehat{V}_{i,j}(\omega_{s})
=
\displaystyle\sum_{k=0}^{K-1}
c_{i,j,k}\,\Psi_k\bigl(p_{s}\bigr),
& \text{PCE},\\[1mm]
\widetilde{V}_{i,j}(\omega_{s})
=
\widetilde{g}_i\bigl(t_j,p_{s}\bigr),
& \text{kernel},
\end{cases}
\label{eq:online_surrogate_samples}
\end{align}
where $s=1,\dots,M_{\mathrm{on}}$.
To quantify the online sampling error, define
\begin{align}
\varepsilon_{M_{\mathrm{on}}}
=
\sqrt{\frac{1}{2M_{\mathrm{on}}}
\log\!\Bigl(\frac{2NT}{\beta}\Bigr)}.
\label{eq:online_epsN_def}
\end{align}
As in Section~\ref{subsec:offline_var_certification}, we assume
$\varepsilon_{M_{\mathrm{on}}}\le\delta$, otherwise the conservative
online quantile level exceeds one and the certificate is not defined in
the present form. We then define the conservative online quantile level
$
p_{M_{\mathrm{on}}}^+
=
1-\delta+\varepsilon_{M_{\mathrm{on}}}.
$
Let
$
v_{1}
\le\cdots\le
v_{M_{\mathrm{on}}}
$
denote the order statistics of
$\{v_{s}\}_{s=1}^{M_{\mathrm{on}}}$.
The conservative online surrogate quantile is then
$\bar{q}_{on,i,j,\delta}
=
v_{s}^{[k_{\mathrm{on}}^+]}$,
where
$
k_{\mathrm{on}}^+
=
\bigl\lceil p_{M_{\mathrm{on}}}^+\,M_{\mathrm{on}}\bigr\rceil,
$
and the corresponding worst-case online surrogate certificate is
$
\bar\rho_{on,\delta}
=
\max_i\max_j\,\bar{q}_{on, i,j,\delta}.
$
By Corollary~\ref{cor:combined_surrogate_certificate}, a sufficient
online decision rule is
$
\bar\rho_{on,\delta}
+\varepsilon_{\mathrm{surr}}
\le 0.
$
If this condition holds, the candidate trajectory is certified on the
grid $\mathcal{T}$. Otherwise, the trajectory is rejected or passed to
a fallback decision layer.
Algorithm~\ref{alg:offline_online} summarizes the complete
offline-online procedure.
\begin{algorithm}[t]
\caption{Offline surrogate construction and online VaR certification}
\label{alg:offline_online}
\begin{algorithmic}[1]
\REQUIRE Candidate trajectory $x_c$ on $\mathcal{T}$, constraint
violation maps $\{g_i\}_{i=1}^{N}$, risk level $\delta\in(0,1)$,
failure probability $\beta\in(0,1)$, surrogate type (PCE or kernel)
\ENSURE Certified online decision based on
$\bar\rho_{\delta, on}+\varepsilon_{\mathrm{surr}}$

\medskip
\textbf{OFFLINE}
\STATE Choose the certification grid $\mathcal{T}=\{t_j\}_{j=1}^{T}$
\STATE Construct a surrogate model:
\begin{itemize}
\item \textbf{PCE:} choose orthonormal basis $\{\Psi_k\}$ and
quadrature nodes $\{(p_q,w_q)\}_{q=1}^{Q}$; compute
coefficients $\{c_{i,j,k}\}$ via \eqref{eq:pce_projection}
\item \textbf{Kernel:} collect training data
$\mathcal{D}_i=\{(z_\ell,y_\ell)\}_{\ell=1}^{L}$
as in Section~\ref{subsec:kernel_extension}; compute coefficients
$\alpha_\ell$
\end{itemize}
\STATE Compute or estimate a uniform surrogate error bound
$\varepsilon_{\mathrm{surr}}$ satisfying
\eqref{eq:surr_uniform_error}

\medskip
\textbf{ONLINE}
\STATE Draw $M_{\mathrm{on}}$ i.i.d.\ sample points
$\omega_{1},\dots,\omega_{M_{\mathrm{on}}}\in\Omega$
\STATE Evaluate surrogate samples $v_{s}$
for all $i=1,\dots,N$, $j=1,\dots,T$, $s=1,\dots,M_{\mathrm{on}}$
via \eqref{eq:online_surrogate_samples}
\STATE Compute $\varepsilon_{M_{\mathrm{on}}}$ from
\eqref{eq:online_epsN_def}
\IF{$\varepsilon_{M_{\mathrm{on}}}>\delta$}
  \STATE Increase $M_{\mathrm{on}}$; the certificate is not defined in
  the present form
\ENDIF
\STATE Set
$p_{M_{\mathrm{on}}}^+=1-\delta+\varepsilon_{M_{\mathrm{on}}}$
\STATE Sort the surrogate samples, compute the conservative quantiles
$\bar{q}_{on,i,j,\delta}$ for all $(i,j)$, and take
their maximum to obtain $\bar\rho_{on,\delta}$
\IF{$\bar\rho_{on,\delta}
+\varepsilon_{\mathrm{surr}}\le 0$}
  \STATE Certify the candidate trajectory and \textbf{accept}
\ELSE
  \STATE Do not certify the candidate trajectory and \textbf{reject}
\ENDIF
\end{algorithmic}
\end{algorithm}

\section{Case Study: Crystallization Population Balance Model}
\label{sec:case_study}
We illustrate the proposed certification framework on a crystallization
process described by a 1-D population balance model (PBM) \cite{ramkrishna2000population}
with uncertain growth and nucleation kinetics. This example is
representative of an expensive nonlinear process model in which direct
Monte Carlo evaluation of time-varying risk requires repeated numerical
solution of a PDE-constrained system and is therefore too costly for
fast online use.

The purpose of this case study is twofold. First, it illustrates surrogate-based risk evaluation for a fixed candidate trajectory under parametric uncertainty. Second, it demonstrates the practical advantage
of surrogate-based risk certification: once the surrogate is constructed
offline, repeated VaR evaluations become significantly faster than running repeated PBM
simulations, while the surrogate based upper bound remains conservative.

In this case study, the certified object is a fixed operating
trajectory, namely the prescribed solubility trajectory
$C^\ast(t_j)$ on the certification grid $\mathcal T$. Thus, in the notation of Section~\ref{sec:problem}, the candidate object
is the deterministic operating profile $x_c(\cdot)=C^\ast(\cdot)$. The
constraint functions are obtained by propagating the PBM under this fixed profile and
then evaluating the resulting quality variables. For brevity, the explicit
dependence on the fixed profile $x_c$ is suppressed in the notation below.

For a fixed realization $p=(k_g,k_b)\in\Pi$ of the uncertain kinetic
parameter vector, the growth and nucleation laws depend on $p$ through
$G(S,p)=G(S;k_g)$, $B(S,p)=B(S;k_b)$. let $n(L,t,p)$ denote the particle number density, where
$L\in[0,L_{\max}]$ is the crystal size coordinate and
$t\in[0,t_f]$ is time. The PBM is
\begin{align}
\frac{\partial n}{\partial t}(L,t,p)
  + \frac{\partial}{\partial L}\!\Bigl(
    G(S(t,p),p)\,n(L,t,p)
  \Bigr)
  = 0,
  \label{eq:pbm_pde}
\end{align}
for $L\in[0,L_{\max}]$ and $t\in[0,t_f]$, with boundary condition
$
G(S(t,p),p)\,n(0,t,p)=B(S(t,p),p)
$
and initial condition
$
n(L,0,p)=n_0(L).
$
Here $G$ and $B$ denote the crystal growth and nucleation rates,
respectively.

The supersaturation is coupled to the particle population through
$
\frac{dC}{dt}(t,p)
  =
  -3\rho_c k_v G(S(t,p),p)\,\mu_2(t,p),
$
where $\rho_c$ is the crystal density, $k_v$ is the volumetric shape
factor, and
$
\mu_j(t,p)
  =
  \int_0^{L_{\max}}L^j n(L,t,p)\,dL
$
denotes the $j$-th moment of the particle size distribution. The
supersaturation is
$
S(t,p)
  =
  \frac{C(t,p)-C^\ast(t)}{C^\ast(t)},
$
where $C^\ast(t)$ is the prescribed solubility profile.

After spatial discretization on a uniform size grid
$L_\ell=\ell\Delta L$, $\ell=0,\dots,N_L$, with
$\Delta L=L_{\max}/N_L$, the state is
$
x(t,p)
  =
  [n_1(t,p),\dots,n_{N_L+1}(t,p),C(t,p)]^\top.
$

We assess the safety of the resulting trajectory under uncertainty using two time-varying quality constraints. First, the volume weighted mean crystal size
$d_{43}(t,\theta)
=
\frac{\mu_4(t,\theta)}{\mu_3(t,\theta)}
$
is required to remain below a prescribed upper bound $d_{43,\max}(t)$. Second, the coefficient of variation of the crystal-size distribution
$
\mathrm{CV}(t,\theta)
=
\frac{\sqrt{\mu_2(t,\theta)/\mu_0(t,\theta)-\left(\mu_1(t,\theta)/\mu_0(t,\theta)\right)^2}}
{\mu_1(t,\theta)/\mu_0(t,\theta)},
$
is required to remain below a prescribed upper bound $\mathrm{CV}_{\max}(t)$. At each
grid point $t_j\in\mathcal T$, define the deterministic constraint function maps
$
r_1(t_j,p)
  = d_{43}(t_j,p)-d_{43,\max}(t_j),
$
$
r_2(t_j,p)
  = \mathrm{CV}(t_j,p)-\mathrm{CV}_{\max}(t_j).
$
Hence, in the notation of Section~\ref{sec:problem}, the corresponding
constraint value function random variables are
$
V_{1,j}=r_1(t_j,P),
\qquad
V_{2,j}=r_2(t_j,P).
$
Equivalently, for $\omega\in\Omega$,
$
V_{1,j}(\omega)=r_1\bigl(t_j,P(\omega)\bigr),
\qquad
V_{2,j}(\omega)=r_2\bigl(t_j,P(\omega)\bigr).
$
Positive values indicate violation, and nonpositive values indicate
satisfaction of the corresponding quality specification.
The uncertainty is purely parametric and time-invariant over each
trajectory. For each sample, a realization
$
P(\omega)=(k_g(\omega),k_b(\omega))
$
of the uncertain parameter vector $P$ is drawn once at
$t=t_0$ from the induced distribution described in
Table~\ref{tab:pbm_settings_extended}, and then kept fixed over the
entire horizon. 
This is equivalent to the general formulation in 
Section~\ref{subsec:pce}, where $P = (k_g, k_b)$ plays the 
role of the uncertain parameter vector and $Z$ is the 
auxiliary variable used to construct the orthonormal basis.
Since the uncertain parameters $(k_g, k_b)$ do not follow a 
standard distribution, we introduce an auxiliary Gaussian random 
vector $Z = (Z_1, Z_2)$ and obtain $(k_g, k_b)$ through the 
transformations
$K_g=\max\{k_g^{\mathrm{nom}}(1+0.2Z_1),\,0.2k_g^{\mathrm{nom}}\}$ and $K_b=\max\{k_b^{\mathrm{nom}}(1+0.3Z_2),\,0.1k_b^{\mathrm{nom}}\}$
as indicated in Table~\ref{tab:pbm_settings_extended}. This 
allows the PCE basis to be constructed in the standard Gaussian 
setting while recovering the correct parameter distribution.

The role of the surrogate is particularly clear in this setting.
Computing the VaR profile directly by Monte Carlo requires solving the
PBM repeatedly for many parameter realizations and then estimating
quantiles at all grid points. During the online phase, the proposed framework
replaces those repeated PBM solving processes by repeated
evaluations of a surrogate constraint model. In the present example, PCE is
well suited because the uncertainty is low-dimensional, parametric, and
enters the constraint value functions smoothly through the kinetics. The kernel surrogate
serves as a data-driven alternative that does not rely on a fixed
polynomial basis and remains applicable when one has only sampled
input-output evaluations.
The numerical settings used in the case study are summarized in
Table~\ref{tab:pbm_settings_extended}. For
$N=2$, $T=301$, $\delta=0.10$, and $\beta=0.05$, the condition
\eqref{eq:epsN_gmall} requires
$
M \ge \frac{1}{2\delta^2}
\log\!\Bigl(\frac{2NT}{\beta}\Bigr)
\approx 505.
$
Moreover, achieving a DKW radius $\varepsilon^*=0.05$ would require
$
M \ge \frac{1}{2(\varepsilon^*)^2}
\log\!\Bigl(\frac{2NT}{\beta}\Bigr)
\approx 2018
$.
The present illustration uses $M_{\mathrm{off}}=200$, which yields
$\varepsilon_{M_{\mathrm{off}}}\approx 0.159>\delta$. Therefore, the
assumptions of Theorem~\ref{thm:offline_certificate} are not satisfied
for a formal distribution-free certificate at the level $\delta=0.10$.
In this case study, the surrogate-based upper curves are
interpreted as conservative empirical risk indicators rather than as
formal finite-sample certificates. A fully rigorous certificate at this
risk level would require a larger certification sample size.
We focus primarily on the constraint function $V_{1,j}$, which quantifies the risk
of crystal coarsening through the $d_{43}$ specification. The same
analysis applies to $V_{2,j}$ for the coefficient of the variation
constraint.

The left panel of Fig.~\ref{fig:pbm_combined} shows snapshots of the
particle size distribution under nominal and worst-case uncertainty
realizations at three selected times. Worst-case realizations shift the
distribution toward larger crystal sizes, driving $d_{43}$ closer to the
time-varying bound $d_{43,\max}(t)$.
Here, the worst-case realization denotes the sampled parameter realization
that produces the largest value of the selected constraint function over the plotted
time interval.
The right panel of Fig.~\ref{fig:pbm_combined} shows the time evolution
of the corresponding quality trajectory and surrogate-based risk
profile. The empirical VaR associated with $V_{1,j}$ increases along
the horizon and reaches its maximum near the end of the batch. The
surrogate-based upper curve
$\bar{q}_{on, 1,j,\delta}+\varepsilon_{\mathrm{surr}}$
remains above the empirical estimate in
this example. Thus the framework provides a fast and conservative risk
indicator for the fixed candidate trajectory.
Figure~\ref{fig:pbm_risk_certificate_v2} shows the analogous result for
the CV-based constraint value $V_{2,j}$. Over part of the horizon, the
uncertainty band overlaps the safety limit, leading to a positive
empirical VaR and therefore a nonzero probability of CV violation,
despite a feasible nominal trajectory. Again, the surrogate-based upper
curve remains conservative in this example.
Finally, Fig.~\ref{fig:var_comparison} compares the VaR profiles
obtained from direct Monte Carlo simulation of the PBM, from the PCE
surrogate, and from the kernel surrogate. The red dashed curve shows the
surrogate-based upper bound
$\bar{q}_{on, 1,j,\delta}+\varepsilon_{\mathrm{surr}}$.
Direct Monte Carlo requires approximately $7164.6$~ms, whereas the PCE
and kernel surrogates require $0.034$~ms and $0.189$~ms, respectively.
Hence the surrogates preserve the relevant risk information while
substantially reducing evaluation time in the online phase.

\begin{figure*}[t]
    \centering
    \begin{subfigure}{0.48\textwidth}
        \centering
        \includegraphics[width=\textwidth]{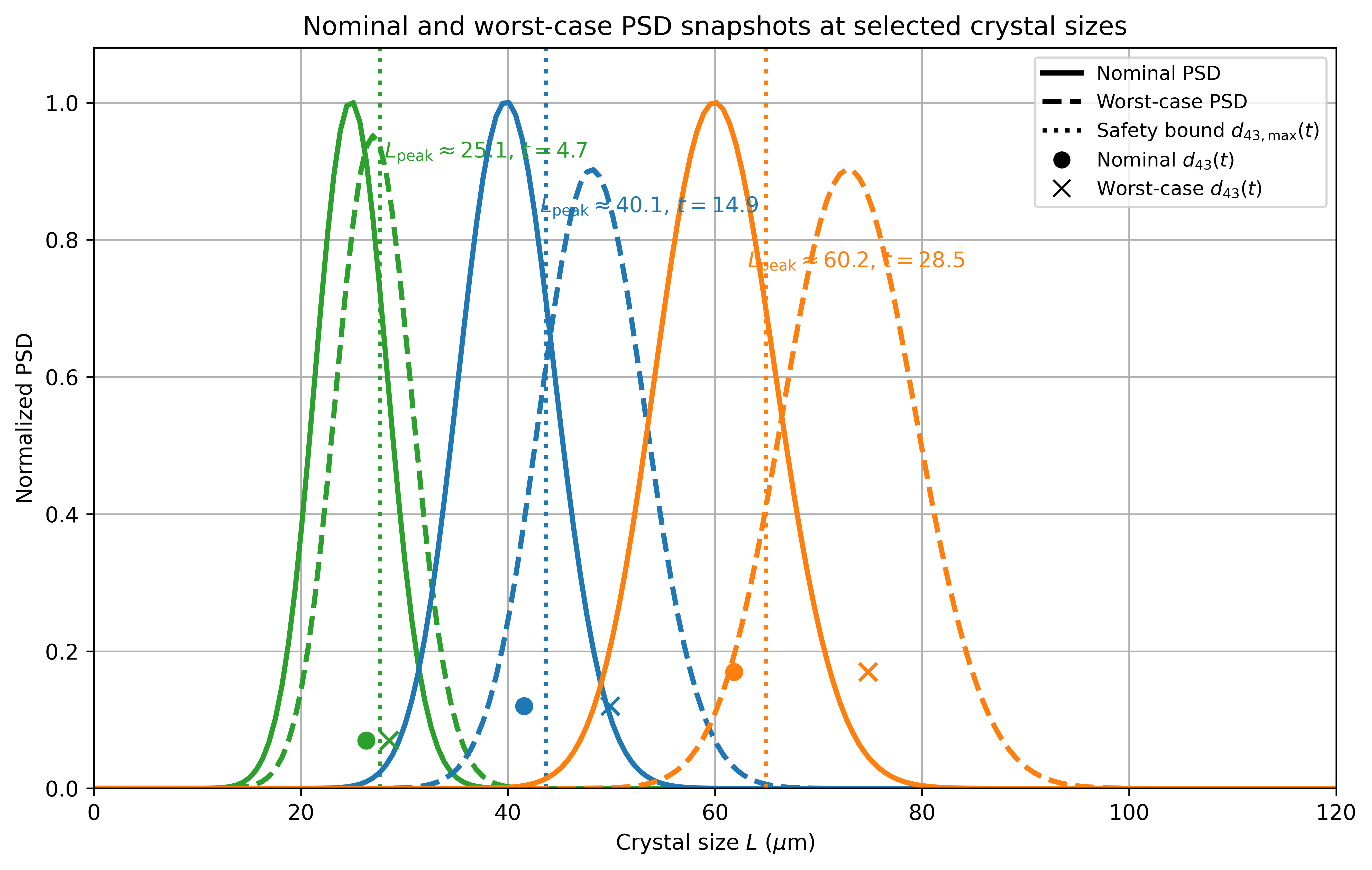}
        \caption{PSD snapshots under nominal and worst-case uncertainty realizations.}
        \label{fig:pbm_psd_snapshots}
    \end{subfigure}
    \hfill
    \begin{subfigure}{0.48\textwidth}
        \centering
        \includegraphics[width=\textwidth]{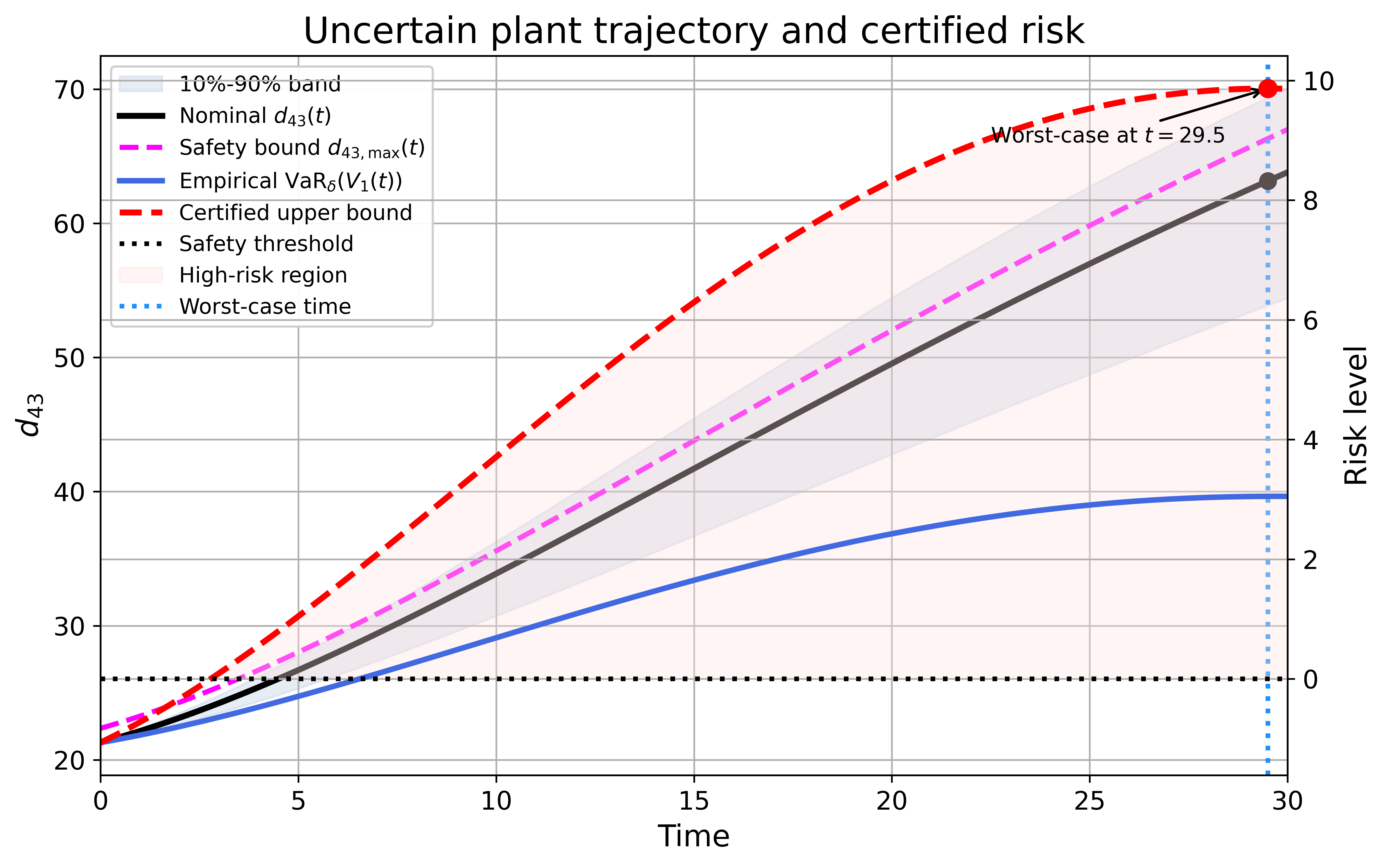}
        \caption{Time evolution of the trajectory and the corresponding certified risk.}
        \label{fig:pbm_risk_certificate}
    \end{subfigure}
    \caption{PBM case study under uncertain growth and nucleation kinetics.
Left: distribution-level effect of uncertainty at selected times.
Right: time-resolved surrogate-based risk assessment based on
$V_{1,j}(\omega)=r_1\bigl(t_j,P(\omega)\bigr)=d_{43}(t_j,P(\omega))-d_{43,\max}(t_j)$.}
    \label{fig:pbm_combined}
\end{figure*}

\begin{figure}[t]
    \centering
    \includegraphics[width=\columnwidth]{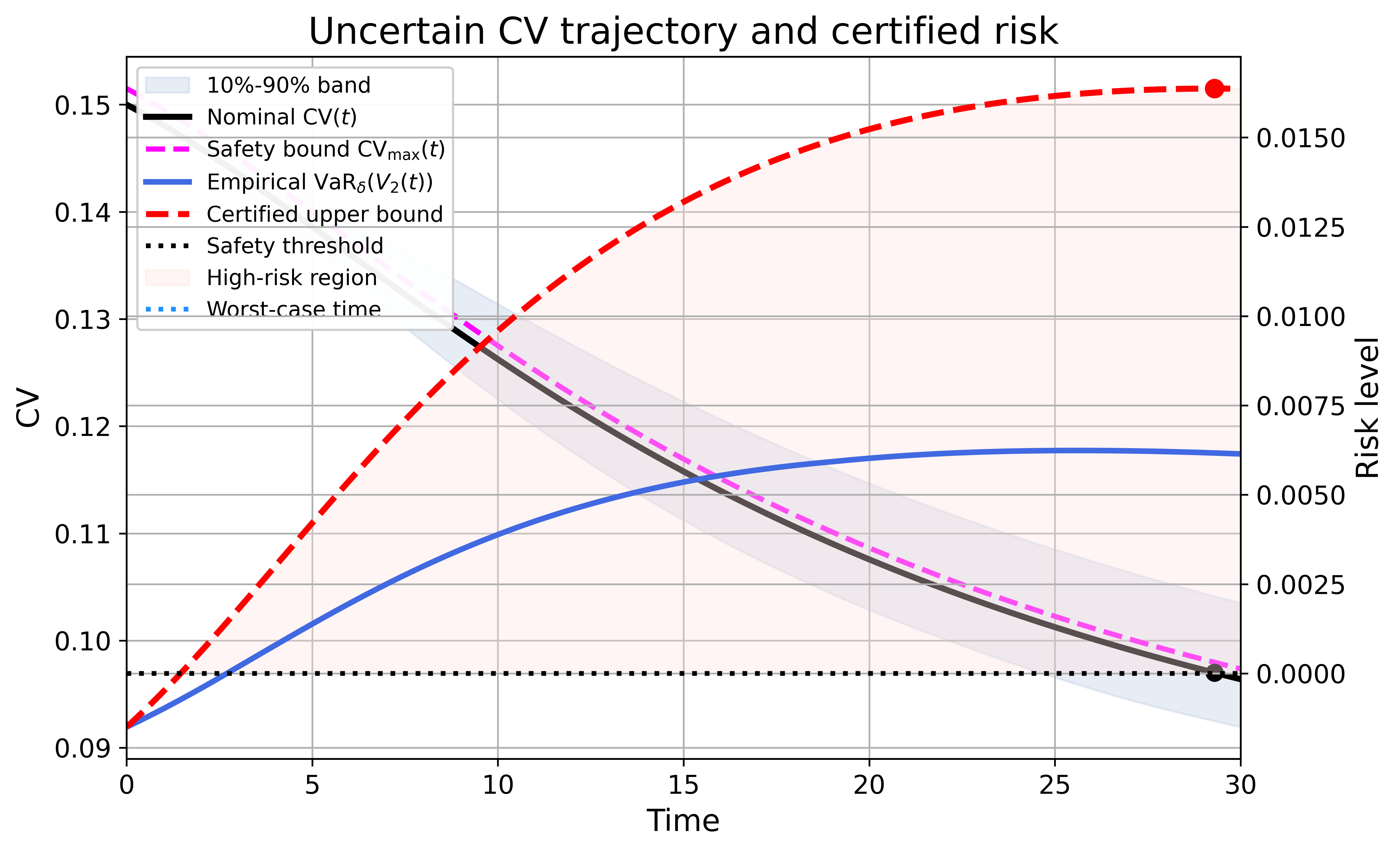}
    \caption{Uncertain trajectory and surrogate-based risk for the CV-based
    constraint value 
    $V_{2,j}(\omega)=r_2\bigl(t_j,P(\omega)\bigr)=\mathrm{CV}(t_j,P(\omega))-\mathrm{CV}_{\max}(t_j)$.
    The left axis shows the nominal trajectory $\mathrm{CV}(t)$,
    the time-varying safety bound $\mathrm{CV}_{\max}(t)$, and the
    $10\%$-$90\%$ uncertainty band. The right axis shows the empirical
    $\mathrm{VaR}_{\delta}(V_{2,j})$ and the conservative finite-sample
    upper confidence bound. Positive values of the risk curves indicate
    potential violation of the CV specification.}
    \label{fig:pbm_risk_certificate_v2}
\end{figure}

\begin{figure}[t]
\centering
\includegraphics[width=\columnwidth]{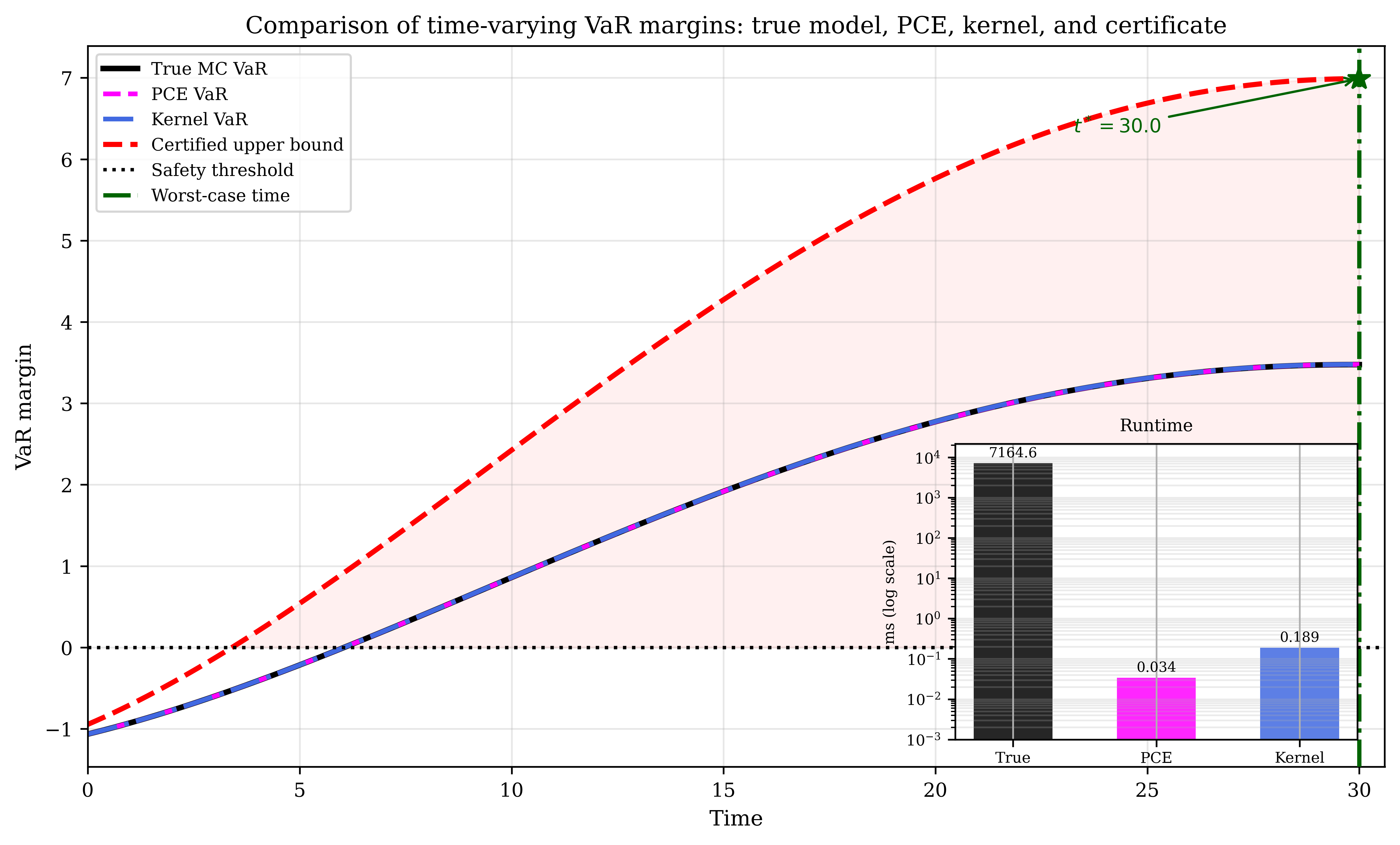}
\caption{Comparison of time-varying VaR margins for $V_{1,j}$
from direct Monte Carlo (original PBM), the PCE surrogate, and the
kernel surrogate. The red dashed curve shows the conservative
surrogate-based upper bound; the vertical dash-dotted line marks the
worst-case time. The inset reports evaluation runtimes.}
\label{fig:var_comparison}
\end{figure}

\begin{table}[t]
\caption{Settings for the crystallization PBM case study.}
\label{tab:pbm_settings_extended}
\centering
\footnotesize
\begin{tabular}{lll}
\hline
\textbf{Item} & \textbf{Symbol} & \textbf{Value / Description} \\
\hline
Time horizon & $t_f$ & $30$ \\
Evaluation grid & $T$ & $301$ \\
Risk level & $\delta$ & $0.10$ \\
Confidence & $\beta$ & $0.05$ \\
Illustrative sample size & $M_{\mathrm{off}}$ & $200$ \\
DKW radius & $\varepsilon_{M_{\mathrm{off}}}$ & $0.159$ \\
Minimum $M$ for $\varepsilon_M\le\delta$ & -- & $505$ \\
Required $M$ for $\varepsilon^*=0.05$ & -- & $2018$ (Cor.~\ref{cor:sample_complexity}) \\
\hline
\textbf{Uncertainty} & & \\
\hline
Growth & $k_g$ &
$Kg$ \\
Nucleation & $k_b$ &
$Kb$ \\
Auxiliary variables & &
$Z_1,Z_2\stackrel{\mathrm{i.i.d.}}{\sim}\mathcal{N}(0,1)$ \\
\hline
\end{tabular}
\end{table}

\section{Conclusion}
\label{sec:conclusion}
This paper presents an offline-online framework for VaR-based
assessment and certification of fixed candidate trajectories under
uncertain time-varying constraints on a prescribed discrete
certification grid. The theoretical results provide a distribution-free
finite-sample upper bound based on the DKW inequality and a
surrogate-transfer bound under a uniform surrogate approximation error.
The PBM example illustrates the computational advantage of
surrogate-based evaluation and shows that the resulting risk curves
track direct Monte Carlo results closely. In the present case study, the
sample size is chosen for illustration and does not satisfy the
finite-sample condition required for a formal certificate at
$\delta=0.10$; establishing such a certificate would require a larger
certification sample budget. Future work will address continuous-time
interpolation guarantees, less conservative confidence bounds, and
operational constructions of surrogate error bounds.

\bibliographystyle{ieeetr}  
\bibliography{ref}
\end{document}